\documentclass[]{article}
\usepackage{relsize}
\usepackage{cite}
\usepackage{color}
\usepackage[dvipsnames]{xcolor}

\usepackage{hyperref}
\hypersetup{
  colorlinks   = true, 
  urlcolor     = blue, 
  linkcolor    = Purple, 
  citecolor   = red 
}
\usepackage{amsmath,amsthm,amssymb}

\usepackage[margin=2.5cm]{geometry}

\usepackage{stmaryrd}

\theoremstyle{definition}
\newtheorem{theorem}{Theorem}
\newtheorem{definition}{{{Definition}}}
\newtheorem{example}{{{Example}}}

\newtheorem{remark}{{{Remark}}}

\newtheorem{corollary}{{{Corollary}}}
\newtheorem{proposition}{{{Proposition}}}
\newtheorem{lemma}{{{Lemma}}}

\newcommand{\C}{\mathcal{C}}

\newcommand{\F}{\mathbb{F}}

\newcommand{\N}{\mathbb{N}}

\newcommand{\wt}{\mathrm{wt}}

\newcommand{\G}{\mathcal{G}}

\newcommand{\dfree}{\mathrm{d_{free}}}

\newcommand{\T}{\mathcal{T}}

\title{Construction of LDPC convolutional codes via \\ difference triangle sets\footnote{This work is an extension of the conference paper ``Construction of Rate $(n-1)/n$ Non-Binary LDPC Convolutional Codes via Difference Triangle Sets'' \cite{alfarano2020construction}, which appears in \emph{2020 IEEE International Symposium on Information Theory (ISIT)}.}}

\author{Gianira N. Alfarano \and Julia Lieb \and Joachim Rosenthal}

\date{Institute of Mathematics, University of Z{u}rich}

\begin{document}
\maketitle

\begin{abstract}
In this paper, a construction of $(n,k,\delta)$ LDPC convolutional codes over arbitrary finite fields, which generalizes the work of Robinson and Bernstein and the later work of Tong is provided. The sets of integers forming a $(k,w)$-(weak) difference triangle set are used as supports of some columns of the sliding parity-check matrix of an $(n,k,\delta)$ convolutional code, where $n\in\N$, $n>k$.  The parameters of the convolutional code are related to the parameters of the underlying difference triangle set. In particular, a relation between the free distance of the code and $w$ is established as well as a relation between the degree of the code and the scope of the difference triangle set. Moreover, we show that some conditions on the weak difference triangle set ensure that the Tanner graph associated to the sliding parity-check matrix of the convolutional code is free from $2\ell$-cycles not satisfying the full rank condition over any finite field. Finally, we relax these conditions and provide a lower bound on the field size, depending on the parity of $\ell$, that is sufficient to still avoid $2\ell$-cycles. This is important for improving the performance of a code and avoiding the presence of low-weight codewords and absorbing sets.
\end{abstract}

\section{Introduction}

In the last three decades, the area of channel coding gained a lot of attention, due to the fact that many researchers were attracted by the practical realization of coding schemes whose performances approach the Shannon limit. This revolution started in 1993 with the invention of turbo codes and their decoding algorithms \cite{berrou1993near}. Only few years later, researchers investigated also low-density parity-check (LDPC) block codes and their message passing decoding algorithm. These codes were discovered to be also capable of capacity-approaching performances.  The class of LDPC block codes was introduced by Gallager \cite{gallager1962low}, in 1962. Their name is due to the fact that they have a parity-check matrix that is sparse. The analysis of LDPC codes attracted many researchers and a lot of work arose in this direction, starting from the papers of Wiberg \cite{wiberg1996codes} and Mackay and Neal \cite{mackay1996near}. Moreover, in \cite{richardson2001capacity,chung2001analysis} analytical tools were introduced to investigate the limits of the performance of the message passing iterative decoding algorithm, suggested by Tanner already in 1981, \cite{tanner1981recursive}.

Similarly to LDPC block codes, one can consider  LDPC convolutional codes. These codes are defined as the (right) kernel of a sparse sliding parity-check matrix, which allows to still use iterative message passing decoding algorithms. Moreover, it was proven that LDPC convolutional codes are practical in different communication applications, see for instance \cite{oswald2002capacity,bates2006termination,bates2005low}.

In the last few years, some attempts to construct binary LDPC convolutional codes were done. Two types of constructions were mainly investigated. The first one exploits the similarity of quasi-cyclic block codes and time-invariant LDPC convolutional codes, \cite{tanner1981error,tanner1987convolutional,tanner2004ldpc}. The second one regards mostly time varying convolutional codes, see for instance \cite{zhou2010cycle,pusane2011deriving,battaglioni2019girth}.

The aim of this paper is to give a combinatorial construction of LDPC convolutional codes suitable for iterative deoding. In fact, contrary to LDPC block codes for which a lot of combinatorial constructions have been derived (see for example \cite{rosenthal2000constructions, kou2000construction, kou2001low, johnson2001regular,vasic2001combinatorial,vasic2004combinatorial}), it is rare to use combinatorial tools for constructing LDPC convolutional codes.

In 1967, Robinson and Bernstein \cite{robinson1967class} used difference triangle sets for the first time to construct binary recurrent codes, which are defined as the (right) kernel of a binary sliding matrix. At that time, the theory of convolutional codes was not developed yet and the polynomial notation was not used, but now, we may regard recurrent codes as a first prototype of convolutional codes. This was the first time that a combinatorial object has been used to construct convolutional codes. Three years later, Tong in \cite{tong1970systematic}, used diffuse difference triangle sets to construct self-orthogonal diffuse convolutional codes, defined by Massey \cite{massey1963threshold}.  The aim of these authors was to construct codes suitable for iterative decoding and their result was an adapted version of binary LDPC convolutional codes. In \cite{alfarano2020construction}, the authors constructed $(n,n-1)_q$ LDPC convolutional codes, whose sliding parity-check matrix is free from $4$ and $6$-cycles not satisfying the so called full rank condition, starting from difference triangle sets. This was a generalization of the work of Robinson and Bernstein, in which difference triangle sets were used to construct convolutional codes over the binary field, that can only avoid $4$-cycles. In 1971, Tong \cite{tong1971} was the first to generalize their construction over $\F_q$, using what we call in this paper weak difference triangle sets. However, his construction is suitable only for limited rate and in a way that the Tanner graph associated to the parity-check matrix of these codes is free only from $4$-cycles.

In this paper, we give a construction of LDPC convolutional codes for arbitrary rates over arbitrary fields, using  difference triangle sets and weak difference triangle sets. In particular, the use of the weak version of these combinatorial objects allows to relax the assumptions required by Robinson, Bernstein and Tong. Indeed, instead of considering sets of nonnegative integers where all the pairwise differences are distinct among all the sets, we may require only that the pairwise differences are distinct in each set. Moreover, we show that using difference triangle sets for this construction produces codes with good distance properties and we provide a bound on the field size that is sufficient to have codes with good distance and to avoid the presence of cycles not satisfying the full rank condition.

The paper is structured as follows. In Section \ref{sec:preliminaries}, we start with some background about convolutional codes. Then we define difference triangle sets (DTSs) and weak difference triangle sets (wDTSs) and their scope. Finally, we introduce LDPC block and LDPC convolutional codes over finite fields of arbitrary size. In Section \ref{sec:construction}, we generalize the construction in \cite{alfarano2020construction} for LDPC convolutional codes to arbitrary rates, starting from a weak difference triangle set. 
We show how the parameters of the constructed convolutional code depend on the properties of the weak difference triangle set. We derive some distance properties of the codes and the exact formula for computing their density. 
Moreover, we show that the free distance and the column distances of convolutional codes constructed from a DTS are in some sense optimal. Finally, we present a construction of LDPC convolutional codes from a wDTS and a lower bound on the corresponding field size such that the free distance is at least 3 and such that the corresponding Tanner graph is free from 4 and 6-cycles not satisfying the FRC.
In Section \ref{sec:fieldsize}, we start with some conditions on the underlying wDTS that ensure that the Tanner graph associated to the sliding parity-check matrix is free from any cycle over any finite field. Afterwards, we give a lower bound for the field size sufficient to ensure that our construction provides a Tanner graph free from $2\ell$-cycles not satisfying the FRC, for $\ell$ odd. If $\ell$ is even, we add some assumptions on the wDTS to be able to derive also in this case a lower bound on the field size. Finally, we modify our construction to be able to relax these conditions on the wDTS, which in turn enlarges the underlying field size.

\section{Preliminaries}\label{sec:preliminaries}
In this section, we provide the background and the known results useful for the rest of the paper.

\subsection{Convolutional Codes}

Let $q$ be a prime power and $\F_q$ be the finite field with $q$ elements. Let $k,n$ be positive integers, with $k< n$ and consider the polynomial ring $\F_q[z]$. An $(n,k)_q$ convolutional code is defined as a submodule $\C$ of $\F_q[z]^n$ of rank $k$, such that there exists a polynomial \emph{generator matrix} $G(z)\in \F_q[z]^{k\times n}$, whose rows generate $\C$, i.e. 
$$\C:=\{u(z)G(z)\mid u(z) \in \F_q[z]^k\}\subseteq \F_q[z]^n.$$
If $G(z)$ is \emph{basic}, i.e., it has a right polynomial inverse, there exists a full row-rank \emph{parity-check matrix} $H(z)\in\F_q[z]^{(n-k)\times n}$ 
such that
$$\C := \{v(z)\in\F_q[z]^n \mid H(z)v(z)^\top={0}\}.$$

Finally, we define the \emph{degree} $\delta$ of the convolutional code $\C$ as the highest degree among the $k\times k$ minors in $G(z)$. When the degree  $\delta$ of the convolutional code is known, we denote $\C$  by $(n,k,\delta)_q$. We omit to specify the field when it is not needed. If $G(z)$ is \emph{reduced}, i.e. the sum of the row degrees of $G(z)$ attains the minimal possible value, then this value is equal to $\delta$.

\begin{lemma}\cite{forney1970,kailath1980}
Let  $H(z)=[h_{i,j}(z)] \in \mathbb F[z]^{(n-k)\times n}$ with row degrees
  $\nu_1,\nu_2, \dots, \nu_{n-k}$ and $[H]_{hr}$ be the highest row degree coefficient matrix defined as the matrix with the $i$-th row  consisting of the coefficients of $z^{\nu_i}$ in the $i$-th row of
  $H(z)$. Then $H(z)$ is reduced if and only  if $[H]_{hr}$ is full row-rank.
\end{lemma}

It is well-known that if $H(z)\in\F[z]^{(n-k)\times n}$ is a basic and reduced parity-check matrix of an $(n,k,\delta)$ convolutional code $\C$, then the sum of the row degrees of $H(z)$ is equal to $\delta$ (see \cite{ro01}) which is also equal to the maximal (polynomial) degree of the full-size minors of $H(z)$. If $H(z)$ is basic but not reduced, then the sum of its row degrees is larger than $\delta$.

There is a natural isomorphism between $\F_q[z]^n$ and $\F_q^n[z]$, which extends to the space of matrices and allows to consider a generator and a parity-check matrix of a convolutional code as polynomials whose coefficients are matrices. In particular, we will consider $H(z)\in\F_q^{(n-k)\times n}[z]$, such that $H(z) = H_0+H_1z +\dots H_{\mu}z^\mu$, with $\mu>0$. With this notation, we can expand the kernel representation $H(z)v(z)^\top$ in the following way:
\begin{equation}\label{eq:kerH}
Hv^\top = \begin{bmatrix}
H_0 & & & &  \\
\vdots & \ddots & & &  \\
H_{\mu} & \cdots & H_0 & &  \\
 & \ddots &  &\ddots &  \\
 & & H_{\mu} &\cdots & H_0 \\
 & & & \ddots & \vdots \\
 & & & & H_\mu
\end{bmatrix}\begin{bmatrix}
v_0 \\ v_1 \\ \vdots \\ v_r
\end{bmatrix}=0,
\end{equation}
where $r = \deg(v)$.
We will refer to the representation of the parity-check matrix of $\C$ in equation \eqref{eq:kerH} as \emph{sliding parity-check matrix}. 

Let $v(z)=\sum_{i=0}^r v_iz^i\in \F_q^n[z]$ be a polynomial vector. We define the \emph{weight} of $v(z)$ as the sum of the Hamming weights of its coefficients,  i.e. $\wt(v(z)) := \sum _{i=0}^r \wt_H(v_i)\in\N_0,$ where $\wt_H(v_i)$ denotes the Hamming weight of $v_i\in\F_q^n$. This definition allows to endow a convolutional code $\C\subseteq\F_q[z]^n$ with a distance. We define the \emph{free distance} of $\C$, denoted by $\dfree(\C)$, as the minimum of the nonzero weights of the codewords in $\C$.
The parameters $\delta$ and $\dfree$ are needed to determine respectively the decoding complexity and the error correction capability of a convolutional code with respect to some decoding algorithm. For this reason, for any given $k$ and $n$ and field size $q$, the aim is to construct convolutional codes with ``small'' degree $\delta$ and  ``large'' free distance $\dfree$. 

For any $j \in \N_0$ we define the \emph{$j$-th column distance of $\C$} as
\begin{align*}
d_j^c(\C)&:= \min \biggl\{\wt\biggl(v_0 +v_1z + \dots +v_jz^j\biggr) \mid v(z) 
\in \C, \ v_0\neq 0\biggr\}\\ 
&=\min \biggl\{\wt\biggl(v_0 + \dots +v_jz^j\biggr) \mid H_j^c[v_0\cdots v_j]^{\top}=0,\ v_0\neq 0 \biggr\},
\end{align*}

where $H_j^c$ is called $j$-\emph{th truncated parity-check matrix} and it is defined for any $j\in \N_0$ as
$$H_j^c:=\begin{bmatrix}
    H_0 & & & \\
    H_1 & H_0 & &  \\
    \vdots & \vdots & \ddots & \\
    H_{j} & H_{j-1} & \cdots & H_0
    \end{bmatrix}\in\F_q^{(j+1)(n-k)\times (j+1) n}.$$
    
We also recall the following well-known result.
\begin{theorem}\cite[Proposition 2.1]{gluesing2006strongly}\label{cd}
Let $\C\subseteq\F_q[z]^n$ be an $(n,k)_q$ convolutional code. Let $d\in\N$. Then the following properties are equivalent.
\begin{enumerate}
\item $d_j^c = d$.
\item None of the first $n$ columns of $H_j^c$ is contained in the span of any other $d -2$ columns and one of the first $n$ columns of  $H_j^c$  is in the span of some other $d-1$ columns of that matrix.

\end{enumerate} 
\end{theorem}


\subsection{Difference Triangle Sets}
A difference triangle set is a collection of sets of integers such that any integer can be written in at most one way as difference of two elements in the same set. Difference triangle sets find application in combinatorics, radio systems, optical orthogonal codes and other areas of mathematics \cite{klove1989bounds,chee1997constructions,chen1992disjoint}. We refer to \cite{colbourn1996difference} for a more detailed treatment. More formally, we define them in the following way, by distinguishing between weak difference triangle sets and difference triangle sets.

\begin{definition}
Let $N,M$ be positive integers. An $(N,M)$-\emph{weak difference triangle set} (wDTS) is a collection of sets $\T:=\{T_1, T_2, \dots, T_N\}$, where for any $1\leq i\leq N$, $T_i:=\{a_{i,j} \mid 1\leq j\leq M\}$ is a set of nonnegative integers such that $a_{i,1} <a_{i,2} < \cdots <a_{i,M}$ and for $1\leq i\leq N$ the differences $a_{i,j}-a_{i,k}$, with $1\leq k <j \leq M$ are distinct. 
If all the differences in all the sets are distinct, we call $\T$ a $(N,M)$-\emph{difference triangle set} (DTS). 
\end{definition}

An important parameter characterizing an $(N,M)$-(w)DTS $\T$ is the \emph{scope} $m(\T)$, which is defined as
$$ m(\T):= \max\{a_{i,M} \mid 1\leq i \leq N\}.$$ 

A very well-studied problem in combinatorics is finding families of $(N,M)$-DTSs with minimum scope. In this work, we will use the sets in a (w)DTS as supports of some columns of a sliding parity-check matrix of a convolutional code. We will then relate the scope of the (w)DTS with the degree of the code. Since we want to minimize the degree of the code, it is evident that the mentioned combinatorial problem plays a crucial role also here.

The name ``difference triangle" is derived from a way of writing the differences inside the sets composing $\T$ in a triangular form .
\begin{example}[wDTS]
Let $\T =\{\{1,2,4,8\}, \{1,3,7,15\}, \{1,5,10,16\}\}$. Then $\T$ is a $(3,4)$-wDTS.

The ``triangles" associated to $\T$ are the following:
$$ 
\begin{array}{cccccc}
     1 & & 2 & & 4  \\
     & 3 & & 6 &  \\
     & & 7 & & 
\end{array}\qquad \quad \begin{array}{cccccc}
     2 & & 4 & & 8  \\
     & 6 & & 12 &  \\
     & & 14 & & 
\end{array}\qquad \quad \begin{array}{cccccc}
     4 & & 5 & & 6  \\
     & 9 & & 11 &  \\
     & & 15 & & 
\end{array}$$
\end{example}

\begin{example}[DTS]\label{dts}
Let $\T =\{\{1,4,16,20\}, \{1,7,12,14\}, \{1,9,18,19\}\}$. Then $\T$ is a $(3,4)$-DTS.

The ``triangles" associated to $\T$ are the following:
$$ 
\begin{array}{cccccc}
     3 & & 12 & & 4  \\
     & 15 & & 16 &  \\
     & & 19 & & 
\end{array} \qquad\quad \begin{array}{cccccc}
     6 & & 5 & & 2  \\
     & 11 &  & 7 &  \\
     & & 13 & &
\end{array} \qquad\quad \begin{array}{cccccc}
     8 & & 9 & & 1 \\
     & 17 &  & 10 &  \\
     & & 18 & &
\end{array}$$

\end{example}


\subsection{LDPC codes over arbitrary finite fields}\label{subsec:NBLDPC}

LDPC codes are known for their performance near the Shannon-limit over the additive white Gaussian noise channel \cite{mackay1996near}. Shortly after they were rediscovered, binary LDPC codes were generalized over arbitrary finite fields. This new construction was first investigated by Davey and Mackay in 1998 in \cite{davey1998low}. In \cite{davey1998monte}, it was observed that LDPC codes defined over a finite field with $q$ elements can have better performances than the binary ones. An LDPC code is defined as the kernel of an $N\times M$ sparse matrix $H$ with entries in $\F_q$. We can associate to $H$ a bipartite graph $\G = (V,E)$, called \emph{Tanner graph}, where $V = V_s \cup V_c$ is the set of vertices. In particular, $V_s=\{v_1,\dots, v_N\}$ is the set of \emph{variable nodes} and $V_c=\{c_1,\dots,c_M\}$ is the set of \emph{check nodes}. $E\subseteq V_s\times V_c$ is the set of edges, with $e_{n,m}=(v_n,c_m)\in E$ if and only if $h_{n,m}\ne 0$. The edge $e_{n,m}$ connecting a check node and a variable node is labelled by $h_{n,m}$, that is the  corresponding \emph{permutation node}. For an even integer $m=2\ell$, we call a simple closed path consisting of $\ell$ check nodes and $\ell$ variable nodes in $\G$ an \emph{$m$-cycle}. The length of the shortest cycle is called the \emph{girth} of $\G$ or girth of $H$. It is proved that higher the girth is, the lower the decoding failure of the bit flipping algorithm is. Moreover, in \cite{poulliat2008design} the authors showed that short cycles in an LDPC code may be harmful if they do not satisfy the so called full rank condition (FRC). This is because if the FRC is not satisfied, the short cycles produce low-weight codewords or they form absorbing sets, \cite{amiri2014analysis}.

Moreover, in \cite{poulliat2008design} and in \cite{amiri2014analysis} it is shown that an $m$-cycle, with $m=2\ell$ in an LDPC code with parity-check matrix $H$ can be represented, up to permutations, by an $\ell\times \ell$ submatrix of $H$ of the form

\begin{equation}\label{eq:cycles}
A=\begin{bmatrix}
a_1 & a_2 & 0 & \cdots &\cdots & 0 \\
0 & a_3 & a_4 & \cdots  &\cdots & \vdots\\
\vdots & & \ddots & & & \vdots \\
\vdots & & & \ddots & & \vdots\\
0 & & & & a_{2\ell-3} & a_{2\ell-2} \\
a_{2\ell} & 0 & \cdots & \cdots &0 & a_{2\ell-1}
\end{bmatrix},
\end{equation}
where $a_i\in\F_q^\ast$. The cycle does not satisfy the FRC if the determinant of $A$ is equal to $0$. In this case, the cycle gives an absorbing set. Hence, it is a common problem to construct LDPC codes in which the shortest cycles satisfy the FRC.

In this work, we are interested in the convolutional counterpart of LDPC block codes, which is given by convolutional codes defined over a finite field $\F_q$ as kernel of a sparse sliding parity-check matrix (here with sparse we mean that in particular each $H_i$ is sparse).

\section{Construction of LDPC convolutional codes}\label{sec:construction}
In this section, we use difference triangle sets to construct LDPC convolutional codes over $\F_q$. The construction was provided for $(n,n-1)_q$ convolutional codes in \cite{alfarano2020construction}. Here, we generalize it for arbitrary $n$ and $k$.

We will construct a sliding parity-check matrix $H$ as in \eqref{eq:kerH}, whose kernel defines a convolutional code. Due to the block structure of $H$, it is enough to consider
\begin{equation}\label{eq:slidingportion}
  \mathcal{H}:= H_{\mu}^c=\begin{bmatrix}
    H_0 & & & \\
    H_1 & H_0 & &  \\
    \vdots & \vdots & \ddots & \\
    H_{\mu} & H_{\mu-1} & \cdots & H_0
    \end{bmatrix},
\end{equation}
since $H$ is then constructed by sliding it. It is easy to see that $H$ does contain a cycle of length $2\ell$ not satisfying the FRC if and only if $\mathcal{H}$ does. 
Assuming that $H_0$ is full rank, we can perform Gaussian elimination on the matrix $$\begin{bmatrix}
H_0\\ H_1\\ \vdots \\H_{\mu}
\end{bmatrix},$$
which results in the block matrix
\begin{equation}\label{eq:matrixH}
\bar{H}=\begin{bmatrix}
    A_0 & | & I_{n-k}  \\
    A_1 & | & 0 \\
    \vdots & & \vdots \\
    A_{\mu} & | & 0
\end{bmatrix},\end{equation}
with $A_i\in \F_q^{(n-k)\times k}$ for $i=1,\hdots, \mu$. 
With abuse of notation, we write $H_0$ for $[A_0|I_{n-k}]$, and $H_i$ for the matrices $[A_i | 0]$.

\begin{remark}\label{rem:H(z)basic}
If we define the matrix $\tilde{H}(z)=\sum_{i=0}^{\mu}A_{i}z^i\in\mathbb F_q[z]^{(n-k)\times k}$, then we obtain that $  H(z)=[\tilde{H}(z)\ I_{n-k}]$ and hence $H(z)$ has a polynomial right inverse, i.e. $H(z)$ is basic.
\end{remark}

Given $n\in\N$, with the following definition we describe how we construct the above mentioned matrix $\bar{H}$ from a $(k, w)$-wDTS, which then will define an $(n,k)_q$ convolutional code.

\begin{definition}\label{def:construction}
Let $k,n$ be positive integers with $n>k$ and $\T:=\{T_1, \dots, T_{k}\}$ be a $(k, w)$-wDTS with scope $m(\T)$. 
Set $\mu=\left\lceil\frac{m(\T)}{n-k}\right\rceil-1$ and define the matrix $\bar{H}\in\F_q^{(\mu+1)(n-k)\times n}$, in which 
the $l$-th column has weight $w$ and support $T_l$, i.e. for any $1\leq i \leq (\mu+1)(n-k)$ and $1\leq l\leq k$, 
$\bar{H}_{i,l}\neq 0$ if and only if $i \in T_l$. 
We say that $\bar{H}$ has support $\T$.
The last $n-k$ columns of $\bar{H}$ are given by $[I_{n-k},0_{n-k},\dots, 0_{n-k}]^\top$.
Derive the matrix $\mathcal{H}$ by ``shifting" the columns of $\bar{H}$ by multiples of $n-k$ and then a sliding matrix $H$ of the form of equation \eqref{eq:kerH}. Finally, define $\C := \ker(\mathcal{H})$ over $\F_q$.
\end{definition}

Observe that if $k=n-1$, we simply get the construction provided in \cite[Definition 4]{alfarano2020construction}.

\begin{proposition}\label{deg}
Let $n,k,w$ be positive integers with $n>k$, $\T$ be a $(k,w)$-wDTS with scope $m(\T)$ and set $\mu=\left\lceil\frac{m(\T)}{n-k}\right\rceil-1$. If $\bar{H}$ has support $\T$, then the corresponding code is an $(n,k,\delta)$ convolutional code with $\mu\leq\delta\leq\mu(n-k)$. Moreover $H_{\mu}$ is full rank if and only if $\delta=\mu(n-k)$. 
\end{proposition}
\begin{proof}
  As the matrix $H(z)$ defined in Remark \ref{rem:H(z)basic} is basic, $\delta$ is the maximal degree of the full-size minors of $H$, which is clearly upper bounded by $\mu(n-k)$. Moreover, any minor formed by a column with degree $\mu$ and suitable columns of the systematic part of $H$ has degree $\mu$, which proves the lower bound. 
  
  If $H_{\mu}$ is full rank, it is equal to $[H]_{hr}$, and $H$ is reduced. Hence, $\delta$ is equal to the sum of the $n-k$ row degrees that are all equal to $\mu$, i.e. $\delta=\mu(n-k)$. If $H_{\mu}$ is not full rank, there are two possible cases. First, if $H_{\mu}$ contains no all-zero row, then $[H]_{hr}=H_{\mu}$ is not full rank, and hence $\delta$ is strictly smaller than the sum of the row degrees which is $\mu(n-k)$. Second, if $H_{\mu}$ contains a row of zeros, then the sum of the row degrees of $H$ is strictly smaller than $\mu(n-k)$ and thus, also $\delta$ is strictly smaller than $\mu(n-k)$.  
\end{proof}

\begin{remark}
If $k<n-k$, i.e. the rate of the code is smaller than $1/2$, then \eqref{eq:matrixH} implies that $H_{\mu}$ cannot be full rank. Moreover, in this case, $[H]_{hr}$ can only be full rank if at least $n-2k$ row degrees of $H$ are zero. 
\end{remark}

\begin{proposition}\label{density}
Let $n,k,w$ be positive integers with $n>k$ and $\T$ be a $(k,w)$-wDTS. Assume $\bar{H}$ has support $\T$ and consider the convolutional code $\C$ constructed as kernel of the sliding parity-check matrix corresponding to $\bar{H}$. If $N$ is the maximal codeword length, i.e. for any codeword $v(z)\in\C$, $\deg(v)+1\leq N/n$, then the sliding parity-check matrix corresponding to $\bar{H}$ has density 
$$\frac{wk+n-k}{(n-k)(\mu n+N)}.$$ 
\end{proposition}

\begin{proof}
To compute the density of a matrix, one has to divide the number of nonzero entries by the total number of entries. The result follows immediately.  
\end{proof}

\begin{theorem}\label{distance}
Let $\C$ be an $(n,k)$ convolutional code with parity-check matrix $H$. Assume that all the columns of 
$\begin{bmatrix}
    A_0^\top & \cdots &    A_{\mu}^\top
\end{bmatrix}^\top$ defined as in \eqref{eq:matrixH}
have weight $w$ and 
denote by $w_j$ the minimal column weight of $\begin{bmatrix}
    A_0^\top & \cdots &    A_{j}^\top
\end{bmatrix}^\top$. For $I\subset\{1,\hdots,(n-k)(\mu+1)\}$ and $J\subset\{1,\hdots,n(\mu+1)\}$ we define $[\mathcal{H}]_{I;J}$ as the submatrix of $\mathcal{H}$ with row indices $I$ and column indices $J$. Assume that for some $\tilde{w}\leq w$ all $I, J$ with $|J|\leq |I|\leq \tilde{w}$ and $j_1:=\min(J)\leq k$ and $I$ containing the indices where column $j_1$ is nonzero, we have that the first column of $[\mathcal{H}]_{I;J}$ is not contained in the span of the other columns of $[\mathcal{H}]_{I;J}$. Then 
\begin{itemize}
    \item[(i)] $\tilde{w}+1\leq\dfree(\C)\leq w+1$,
    \item[(ii)] $\min(w_j,\tilde{w})+1\leq d_j^c(\C)\leq w_j+1$.
\end{itemize}
\end{theorem}

\begin{proof}
(i) Without loss of generality, we can assume that the first entry in the first row of $H_0$ is nonzero. Denote the first column of $\mathcal{H}$ by $[h_{1,1},\hdots,h_{1,(n-k)\mu}]^\top$.
Then, $v(z)=\sum_{i=0}^r v_iz^i$ with 
\begin{align*}
    v_0&=[1\ 0\cdots 0\ -h_{1,1}\cdots\ -h_{1,(n-k)}]\quad \textnormal{ and } \\
    v_i&=[0\ 0\cdots 0\ -h_{1,(n-k)i+1}\cdots\ -h_{1,(n-k)(i-1)}],
\end{align*}
for $i\geq 1$ is a codeword with $\wt(v(z))=w+1$ as the weight of the first column of $\mathcal{H}$ is equal to $w$. Hence $\dfree\leq w+1$.

Assume by contradiction that there exists a codeword $v(z)\neq 0$ with weight $d\leq \tilde{w}$. We can assume that $v_0\neq 0$, i.e. there exists $i\in\{1,\ldots,n\}$ with $v_{0,i}\neq 0$. We know that $\mathcal{H} v^{\top}=0$ and from \eqref{eq:matrixH} we obtain that there exists $j\in\{1,\ldots,n\}$ with $j\neq i$ and $v_{0,j}\neq 0$ and we can assume that $i\leq k$. 

Now, we consider the homogeneous system of linear equations given by $\mathcal{H} v^{\top}=0$ and we only take the rows, i.e. equations, where column $i$ of $\mathcal{H}$ has nonzero entries. Moreover, we define $\tilde{v}\in\mathbb F^d$ as the vector consisting of the nonzero components of $v_0, v_1, \ldots, v_{\deg(v)}$. We end up with a system of equations of the form  $[\mathcal{H}]_{I;J}\tilde{v}^{\top}=0$ where $[\mathcal{H}]_{I;J}$ fulfills the assumptions stated in the theorem. But this is a contradiction as $\tilde{v}^{\top}$ has all components nonzero and therefore $[\mathcal{H}]_{I;J}\tilde{v}^{\top}=0$ implies that the first column of $[\mathcal{H}]_{I;J}$ is contained in the span of the other columns of this matrix.\\
(ii) The result follows from Theorem \ref{cd} with an analogue reasoning as in part (i).  
\end{proof}

\begin{remark}\label{rem:dist}
With the assumptions of Theorem \ref{distance}, if $\tilde{w}=w$, one has $d_j^c=\dfree$ for $j\geq\mu$. Moreover, if $\bar{H}$ has support $\T$, one achieves higher column distances (especially for small $j$) if the elements of $\mathcal{T}$ are small.
\end{remark}

\begin{corollary}\label{s}
If $\mathcal{T}$ is a $(k,w)$-DTS and $\C$ is an $(n,k)$ convolutional code constructed from $\T$ as in Definition \ref{def:construction}, then one has that:
\begin{itemize}
    \item[(i)] $\dfree(\C)= w+1$,
    \item[(ii)] $ d_j^c(\C)= w_j+1$.
\end{itemize}
\end{corollary}

\begin{proof}
As already mentioned in \cite{robinson1967class}, matrices $\mathcal{H}$ constructed from a DTS have the property that for every pair of columns, their supports intersect at most once. Since $[\mathcal{H}]_{I;J}$ as defined in Theorem \ref{distance} has the property that all entries in the first column are non-zero, all other columns have at most one non-zero entry. But this implies that the first column cannot be in the span of the other columns and thus, the requirements of Theorem \ref{distance} are fulfilled for $\tilde{w}=w$, which proves the corollary.  
\end{proof}

\begin{remark}\label{nk}
If $n-k>1$, it is not necessary to have a DTS to obtain that all columns of $\mathcal{H}$ intersect at most once since one only has to consider shifts of columns by multiples of $n-k$. Therefore, we still need to consider a set $\T = \{T_1,\dots, T_k\}$ such that all the differences $a_{i_1,j_1}-a_{i_1,s_1}$ and $a_{i_2,j_2}-a_{i_2,s_2}$ for $i_1\neq i_2$ are different, i.e. two differences coming from different triangles of $\T$ have always to be different, but $a_{i,j_1}-a_{i,s_1}$ and $a_{i,j_2}-a_{i,s_2}$, i.e. differences coming from the same triangle, only have to be different if $(n-k)\mid (a_{i,j_1}-a_{i,j_2})$.
\end{remark}

\begin{example}
Consider $n=3$, $k=1$ and $T_1=\{1,2,3\}$. It holds $2-1=3-2$ but since $3-2$ is not divisible by $n-k=2$, this does not matter and we still get that all columns of $\mathcal{H}$ intersect at most once. For example for $\mu=1$, we get $$\mathcal{H}=\left[\begin{matrix}1 &  1& 0 &0 & 0 & 0\\ 1 & 0 & 1 & 0 & 0 & 0\\ 1 & 0 & 0 & 1 & 1 & 0\\0 & 0 & 0 & 1 & 0 & 1\end{matrix}\right].$$
\end{example}

From Corollary \ref{s} we know that if we use a DTS to construct the parity-check matrix of the code, then the values of the nonzero entries are not important to achieve good distance properties. In the following, we present a construction that achieves also quite large distances if one takes the sets in a wDTS as support sets for the columns of the non-systematic part of $\bar{H}$. Moreover, in Section \ref{sec:fieldsize}, we show that this construction ensures that the Tanner graph associated to $H$ is free from cycles of arbitrary length not satisfying the FRC if the size of the underlying field is sufficiently large and the wDTS fulfills some additional properties.

\begin{definition}\label{Construction}
Let $k,n$ be positive integers with $n>k$ and $\T:=\{T_1, \dots, T_{k}\}$ be a $(k, w)$-wDTS with scope $m(\T)$. 
Set $\mu=\left\lceil\frac{m(\T)}{n-k}\right\rceil-1$ and let $\alpha$ be a primitive element for $\F_q$, so that every non-zero element of $\F_q$ can be written as power of $\alpha$. For any $1\leq i \leq (\mu+1)(n-k)$, $1\leq l\leq k$, define
$$
\bar{H}^\T_{i,l} := \begin{cases}\alpha^{il} & \text{ if } i \in T_l \\
0 & \text{ otherwise}
\end{cases}.$$
 Obtain the matrix $\mathcal{H}^\T$ by ``shifting" the columns of $\bar{H}^\T$ by multiples of $n-k$ and then a sliding matrix $H^\T$ of the form of equation \eqref{eq:kerH}. Finally, define $\C^\T := \ker(\mathcal{H}^\T)$ over $\F_q$.


\end{definition}

\begin{example}\label{ex:Construction}
Let $\F_q:=\{0,1,\alpha, \dots, \alpha^{q-2}\}$ 
and $\T$ be a $(2,3)$-wDTS, such that $T_1:=\{1,2,6\}$ and $T_2:=\{1,2,4\}$. 
Then, with the notation above,
$$\bar{H}^\T= \begin{bmatrix}
\alpha & \alpha^{2} & 1 \\
\alpha^{2} & \alpha^{4} & 0 \\
0 & 0 & 0 \\
0 & \alpha^8 & 0 \\
0 & 0 & 0 \\
\alpha^6 & 0 & 0 
\end{bmatrix},$$
which leads to the following sliding matrix.

$$\mathcal{H}^\T= \left[\begin{array}{cccccccccccccccccc}
\alpha & \alpha^{2} & 1 & & & & & & & &  \\

\alpha^{2} & \alpha^{4} & 0 & \alpha & \alpha^2 & 1 & & & & & & & & & & & \\

0 & 0 & 0  & \alpha^{2} & \alpha^{4} & 0  & \alpha & \alpha^2 & 1 & & & & & & & & \\

0 & \alpha^8 & 0 & 0 & 0 & 0  & \alpha^{2} & \alpha^{4} & 0  & \alpha & \alpha^2 & 1 & & & & & \\

0 & 0 & 0 &  0 & \alpha^8 & 0&  0 & 0 & 0  & \alpha^{2} & \alpha^{4} & 0  & \alpha & \alpha^2 & 1 & & &  \\

\alpha^6 & 0 & 0 & 0 & 0 & 0 &  0 & \alpha^8 & 0 & 0 & 0 & 0  & \alpha^{2} & \alpha^{4} & 0  & \alpha & \alpha^2 & 1 \\
\end{array}\right].$$

The code constructed here is a $(3,2)_q$ convolutional code.
In this example, one has $d_0^c=2$, $d_1^c=d_2^c=d_3^c=d_4^c=3$ and $d_5=\dfree=4$.
\end{example}








The next theorem is a generalization of \cite[Theorem 12]{alfarano2020construction} to any rate. 

\begin{theorem}\label{thm:2x2minors}
Let $w, n, k$ be positive integers with $n>k$ and $\T$ be a $(k,w)$-wDTS with scope $m(\T)$ and  $q>(\mu+1)(n-k)(k-1)+1=\lceil\frac{m(\T)}{n-k}\rceil(n-k)(k-1)+1$. Let $\C^\T$ be the $(n,k)_q$ convolutional code defined from $\T$, as defined in Definition \ref{Construction} and consider $\mathcal{H}^\T$ as in \eqref{eq:slidingportion}.
Then, all the $2\times 2$ minors in $\mathcal{H}^\T$ that are non-trivially zero are non-zero.
\end{theorem}
\begin{proof}
The only $2\times 2$ minors to check are the ones of the form $\begin{vmatrix}a_1 & a_2\\ a_3 & a_4
\end{vmatrix}$. By definition of wDTS, the support of any column of $\mathcal{H}^\T$ intersects the support of its shift at most once. This ensures that the columns of all these minors are the shift of two different columns of $\bar{H}^\T$. Moreover, all the elements in the minor are powers of $\alpha$. In particular, let $1\leq i,r \leq (\mu+1)(n-k)$, $1\leq j,\ell \leq k$ (note that $j<\ell$ or $\ell<j$ according to which columns from $\bar{H}^\T$ are involved in the shifts). Hence we have that:
\begin{align*}
& \begin{vmatrix}a_1 & a_2\\ a_3 & a_4
 \end{vmatrix}  =
 \begin{vmatrix}\alpha^{ij} & \alpha^{m\ell}\\ \alpha^{(i+r)j} & \alpha^{(m+r)\ell}
\end{vmatrix}  = \\
&\alpha^{ij}\alpha^{(m+r)\ell} - \alpha^{m\ell}\alpha^{(i+r)j} =
\alpha^{ij + m\ell}(\alpha^{r\ell}-\alpha^{rj})
\end{align*}
which is $0$ if and only if $r\ell = rj \mod (q-1)$. Since it holds that $0\leq j < \ell \leq k$ or $0\leq \ell < j \leq k$  and $1\leq r \leq (\mu+1)(n-k)$, this cannot happen.  
\end{proof}

The following theorem is a generalization of \cite[Theorem 13]{alfarano2020construction} for any rate. However, in the proof in \cite{alfarano2020construction} there is a computation mistake, hence we put the correct version below. 

\begin{theorem}\label{thm:3x3minors}
Let $w, n, k$ be positive integers with $n>k$ and $\T$ be a $(k,w)$-wDTS with scope $m(\T)$, $w\geq 3$. Let $\C^\T$ be the $(n,k)_q$ convolutional code defined from $\T$, as in Definition \ref{Construction} with $\mathcal{H}^\T$ as defined in \eqref{eq:slidingportion} and assume that  $(\mu+1)(n-k)>2$. Assume also that $q=p^N$, where $p>2$ and $$N>(\mu+1) (n-k)(k-1)= \Big\lceil \frac{m(\T)}{n-k}\Big\rceil(n-k)(k-1).$$ Then, all the $3\times 3$ minors in $\mathcal{H}^\T$ that are non-trivially zero are non-zero.
\end{theorem}
\begin{proof}
We need to distinguish different cases. \newline
\underline{\textbf{Case I}}. The $3\times 3$ minors are of the form $$\begin{vmatrix}a_1 & a_2 & a_3\\ a_4 & a_5 & a_6 \\ a_7 & a_8 & a_9
\end{vmatrix},$$ with $a_i \ne 0$ for any $i$. As we observed in Theorem \ref{thm:2x2minors}, in this case all the columns are shifts of three different columns from $\bar{H}^\T$, since each column can intersect any of its shifts at most once. Observe that we can write a minor of this form as 
\begin{gather*}
    \begin{vmatrix}a_1 & a_2 & a_3\\ a_4 & a_5 & a_6 \\ a_7 & a_8 & a_9
\end{vmatrix} = \begin{vmatrix}
\alpha^{ij} & \alpha^{lu} & \alpha^{tm} \\
\alpha^{(i+r)j} & \alpha^{(l+r)u} & \alpha^{(t+r)m}\\
\alpha^{(i+r+s)j} & \alpha^{(l+r+s)u} & \alpha^{(t+r+s)m}\\
\end{vmatrix},
\end{gather*}
where $1\leq i,l,t \leq (\mu+1)(n-k)$, $r,s \in\mathbb{Z}$ are possibly negative, with $r\ne s$, and $1\leq j,u,m\leq k$ representing the index of the column from which the selected element comes from (or if the selected elements belongs to the shift of some column, $j,u,m$ are still the indexes of the original column). Due to symmetry in this case we can assume $r,s\in\mathbb N$ and $1\leq i,l,t \leq (\mu+1)(n-k)-3$. Moreover, $-(\mu+1)(n-k) +1\leq i+r, l+r, t+r\leq (\mu+1)(n-k) - 1$ and $-(\mu+1)(n-k)\leq i+r+s, l+r+s, t+r+s \leq (\mu+1)(n-k)$.
This determinant is $0$ if and only if 
\begin{gather}
\alpha^{ru+rm+sm} + \alpha^{rm+rj+sj}+ \alpha^{rj+ru+sk}=\nonumber\\
\alpha^{ru+rj+sj}+ \alpha^{rj+rm+sm}+ \alpha^{ru+rm+sk}\label{eq:determinant3x3}.    
\end{gather}
Without loss of generality we can assume that $j<u<m$ and it turns out that the maximum exponent in equation \eqref{eq:determinant3x3} is $ru+rm+sm$ while the minimum is $ru + rj + sj$. Let $M:=ru+rm+sm - (ru + rj + sj)$. It is not difficult to see that the maximum value for $M$ is $((\mu+1)(n-k)-1)(k-1)$ hence this determinant can not be zero because $\alpha$ is a primitive element for $\F_q$ and, by assumption, $q=p^N$, where $N>M$.

\underline{\textbf{Case II}}. The $3\times 3$ minors are of the form $$\begin{vmatrix}a_1 & a_2 & 0\\  a_3 & a_4 & a_5\\ a_6 & 0 & a_7
\end{vmatrix}.$$
As in the first case, we can assume that the minor is given by 
\begin{gather*}
    \begin{vmatrix}
\alpha^{ij} & \alpha^{lu} & 0 \\
\alpha^{(i+r)j} & \alpha^{(l+r)u} & \alpha^{(t+r)m}\\
\alpha^{(i+r+s)j} & 0 & \alpha^{(t+r+s)m}\\
\end{vmatrix},
\end{gather*}
with the same bounds on the variables as before.
But, in this case $j\ne u,m$ but $u$ can be equal to $m$. Indeed, the first column intersects the other two in two places, which means that they are not all shifts of the same column. However, the second and third ones can belong to the same column.
This determinant is $0$ when $\alpha^{ru+sm}+\alpha^{rj+sj}- \alpha^{rm+sm}=0$. In this case, according to the different possibilities for $j,u,m$ and $r,s$ we check the maximum and the minimum exponent. We present here only the worst case for the field size, which is obtained when $j<u<m$, $r<0$. We see that the minimum exponent is $rj+sj$ and the maximum is $rj+sm$. We consider $M:=rj+sm-rj-sj$ and we check what is the maximum value that $M$ can reach. It is not difficult to see that this is $(\mu+1)(n-k)(k-1)$. When $p=p^N$, with $N>M$, the considered determinant is never $0$.

\underline{\textbf{Case III}}. The $3\times 3$ minors are of the form $$\begin{vmatrix}a_1 & a_2 & a_3\\ a_4 & a_5 & a_6 \\ a_7 & a_8 & 0
\end{vmatrix},$$ with $a_i \ne 0$ for any $i$. We can assume that, the minor is given by 
\begin{gather*}
    \begin{vmatrix}
\alpha^{ij} & \alpha^{lu} & \alpha^{tm} \\
\alpha^{(i+r)j} & \alpha^{(l+r)u} & \alpha^{(t+r)m}\\
\alpha^{(i+r+s)j} &\alpha^{(l+r+s)u} & 0 \\
\end{vmatrix},
\end{gather*}
with the same bounds on the variables as in previous cases.
However, this time $1\leq j<u<m\leq k$. After some straightforward computations, we get that this determinant is 0 if and only if \begin{gather}\label{eq:determinant3x3minusone}
\alpha^{rm+rj+sj}+ \alpha^{rj+ru+su}=\alpha^{ru+rj+sj}+ \alpha^{ru+rm+su}.
\end{gather} 
In the worst case, consider $M:=ru+rj+su - (rm+rj+sj) = r(u-m)+s(u-j)$ with $r<0$. We immediately see that the maximum value that $M$ can reach is  $(\mu+1)(n-k)(k-2)+1$, hence this determinant can not be zero because $\alpha$ is a primitive element for $\F_q$ and, by assumption, $q=p^N$, where $N>M$.

\underline{\textbf{Case IV}}. The $3\times 3$ minors are of the form $$\begin{vmatrix}a_1 & a_2 & 0\\ 0 & a_3 & a_4 \\ a_6 & 0 & a_5
\end{vmatrix}.$$ 
In this case, we can have that the three considered columns come from different shifts of the same one, hence we allow that some (or all) among $j,u,m$ are equal. 
Arguing as before, we notice that these minors are given by
\begin{gather*}
    \begin{vmatrix}
\alpha^{ij} & \alpha^{lu} & 0 \\
0 & \alpha^{(l+r)u} & \alpha^{(t+r)m}\\
\alpha^{(i+r+s)j} & 0 & \alpha^{(t+r+s)m}\\
\end{vmatrix} =\\ \alpha^{ij+lu+tm+rm}(\alpha^{ru+sm}+\alpha^{rj+sj}).
\end{gather*}
This determinant is $0$ whenever $r(u-j) + s(m-j) - (q-1)/2=0 \mod (q-1)$. 
Analyzing all the possibilities we can have according to $r,s$ being negative or positive and $j,u,m$ being equal or different, after some computations, we obtain that, whenever $q>2(k-1)((\mu+1)(n-k)-1)+1$, the considered determinant is never $0$. 
And this is the case for our field size assumption.   
\end{proof}

Observe that Case IV of Theorem \ref{thm:3x3minors} corresponds to the lower bound for the field size sufficient to avoid the presence of $6$-cycles not satisfying the FRC. Hence, we have the following result.

\begin{corollary}
Let $\mathcal{C}^\T$ be an $(n,k)$ convolutional code constructed from a $(k,w)$ wDTS $\T$ and satisfying the conditions of Theorem \ref{thm:2x2minors} and Theorem \ref{thm:3x3minors}. Then, $d_{free}(\mathcal{C}^\T)\geq 3$ and the code is free from $4$ and $6$-cycles not satisfying the FRC.
\end{corollary}

\begin{remark}
If $\mathcal{C}^\T$ is an $(n,k)$ convolutional code constructed from a $(k,w)$ wDTS $\T$ and satisfying the conditions of Theorem \ref{thm:2x2minors} and Theorem \ref{thm:3x3minors}, such that $H_{\mu}$ has no zero row and $n-k\leq \min\{3,k\}$, then, it follows from Proposition \ref{deg} that $\delta=\mu(n-k)$.
\end{remark}

\begin{example}
Consider the $(3,2)_q$ code constructed in Example \ref{ex:Construction}. Note that $\mu = 5$, hence, for $q>11$ we can avoid all the $6$-cycles not satisfying the FRC (Case IV of Theorem \ref{thm:3x3minors}).
\end{example}


\section{Excluding 2$\ell$-cycles not satisfying the FRC}\label{sec:fieldsize}

In this section, we give some conditions that ensure that the Tanner graph associated to the sliding parity-check matrix of a convolutional code constructed via a difference triangle set is free of $2\ell$-cycles not satisfying the FRC.

First of all we recall from Subsection \ref{subsec:NBLDPC} that a $2\ell$-cycle can be represented by an $\ell \times \ell$ submatrix of $\mathcal{H}$ that up to column and row permutations is of the form 

\begin{equation}\label{eq:cycleA}
A=\begin{bmatrix}
a_1 & a_2 & 0 & \cdots &\cdots & 0 \\
0 & a_3 & a_4 & \cdots  &\cdots & \vdots\\
\vdots & & \ddots & & & \vdots \\
\vdots & & & \ddots & & \vdots\\
0 & & & & a_{2\ell-3} & a_{2\ell-2} \\
a_{2\ell} & 0 & \cdots & \cdots &0 & a_{2\ell-1}
\end{bmatrix},
\end{equation}
where $a_i\in\F_q^\ast$. 

\begin{remark}\label{rem:cyclelength}
Observe that 
$$\left[\begin{matrix}A_0 & & \\ \vdots & \ddots & \\A_{\mu} & \cdots & A_0\end{matrix}\right]\in\mathbb F^{(\mu+1)(n-k)\times(\mu+1)k},$$ 
hence it is clear that the Tanner graph associated to $H$ can only contain $2\ell$-cycles for $$\ell\leq\min\{(\mu+1)(n-k), (\mu+1)k\}.$$ 
\end{remark}

At first, we will investigate conditions on the wDTS used to construct the convolutional code that ensure that the associated Tanner graph contains no cycles at all independently of the nonzero values of the sliding parity-check matrix and hence also independently of the underlying finite field.

\begin{proposition}\label{dif}
If $\mathcal{C}$ is an $(n,k)$ convolutional code whose parity-check matrix has support $\T$ where $\T$ is a $(k,w)$-wDTS with the property that none of the differences $a_{i,j}-a_{i,m}$ for $1\leq i\leq k$ and $1\leq m<j \leq w$ is divisible by $n-k$, then each pair of columns that is next to each other in $A$ as in \eqref{eq:cycleA} consists of shifts of different columns of $\bar{H}$. In particular, at most $\lfloor\frac{\ell}{2}\rfloor$ columns of $A$ can be shifts of the same column of $\bar{H}$.
\end{proposition}

\begin{proof}
The fact that none of the differences in the set is divisible by $n-k$ implies that the support of any column of $\bar{H}$ does not intersect the support of any of its shifts (by multiples of $n-k$). Since the supports of neighbouring columns of $A$ intersect, they have to be shifts of different columns of $\bar{H}$.   
\end{proof}

\begin{corollary}
If $\mathcal{C}$ is an $(n,k)$ convolutional code whose parity-check matrix has support $\T$ where $\T$ is a $(k,w)$-wDTS with the property that $T_1=\cdots=T_k$ and none of the differences $a_{1,j}-a_{1,m}$ for $1\leq m<j \leq w$ is divisible by $n-k$, then the Tanner graph associated to the parity-check matrix $H$ of $\mathcal{C}$ is free from cycles of any size (over every base field) not satisfying the FRC..
\end{corollary}

\begin{theorem}
Assume that $\mathcal{C}$ is an $(n,k)$ convolutional code constructed from an $(k,w)$-DTS $\T$ with $a_{i,1}=1$ for all $1\leq i\leq k$, where $(n-k)$ does not divide any of the nonzero differences $a_{i_1,j}-a_{i_2,m}$ for $1\leq i_1, i_2\leq k$ and $1\leq m, j \leq w$. 
Then, the Tanner graph associated to the parity-check matrix $H$ of $\mathcal{C}$ is free from cycles of any size (over every base field) not satisfying the FRC.
\end{theorem}

\begin{proof}
Assume by contradiction that $\mathcal{H}$ contains up to permutations a submatrix $A$ of the form \eqref{eq:cycleA}.
As the supports of the first two columns of $A$ intersect, they have to be shifts of different columns of $\bar{H}$. The supports of such shifts can only intersect once and the entries of this intersection come from the first row of $\bar{H}$. Applying the same reasoning to the intersection of the supports of the second and third column of $A$, implies that $a_2$ and $a_3$ in $A$ both come from the first row of $\bar{H}$ which is not possible. This shows the result.  
\end{proof}


\begin{example}
Consider the $(2,3)$-DTS $\T=\{T_1,T_2\}$ with $T_1=\{1,2,5\}$ and $T_2=\{1,3,9\}$. The set of all occurring nonzero differences $a_{i_1,j}-a_{i_2,m}$ is $\{1,2,3,4,6,7,8\}$, i.e. none of them is divisible by $5$. Hence the matrix $H(z)=H_0+H_1z$ with $H_0=[\bar{H}_0\ I_5]$ and $H_1=[\bar{H}_1\ 0_5]$, where 
\begin{gather*}
    \bar{H}_0=\left[\begin{matrix}1 & 1\\1 & 0\\ 0 & 1\\ 0 & 0\\ 1 & 0\end{matrix}\right], \qquad  \bar{H}_1=\left[\begin{matrix}0 & 0\\0 & 0\\ 0 & 0\\ 0 & 1\\ 0 & 0\end{matrix}\right]
\end{gather*}
and $$[H]_{hr}=\left[\begin{matrix} 1 & 1 & 1 & 0 & 0 & 0 & 0\\ 1 & 0& 0 & 1 & 0 & 0 & 0\\ 0 & 1 & 0 & 0 & 1 & 0 & 0\\ 0 & 1 & 0 & 0 & 0 & 0 & 0\\1 & 0 & 0 & 0 & 0  & 0 & 1\end{matrix}\right]$$ full rank, i.e. $\delta=1$, is the parity-check matrix of an $(7,2,1)_q$ convolutional code that is free of cycles of any size for any prime power $q$.
\end{example}

Next, we want to relax the conditions on the wDTS used for construction of the convolutional code but still exclude cycles in the Tanner graph of the sliding parity-check matrix that do not fulfill the FRC by using the construction from Definition \ref{Construction} and considering sufficiently large field sizes.

To ensure that the considered cycle does not satisfy the FRC, we have to guarantee that $\det A\ne 0$ as an element of $\F_q$. It is easy to check that $$\det A = \prod_{\substack{i=1 \\ i \textnormal{ odd}}}^{2\ell} a_i \pm \prod_{\substack{i=1 \\ i \textnormal{ even}}}^{2\ell} a_i.$$

Let $\T$ be a $(k,w)$-wDTS and let $\C^\T$ be the convolutional code defined from $\T$, with $\mathcal{H}^\T$ as defined in \eqref{eq:slidingportion}. Each matrix representation $A$ of a $2\ell$-cycle comes from selecting $\ell$ rows and $\ell$ columns of $\mathcal{H}^\T$. Moreover, in each column of $A$, exactly two positions are non-zero. 
Let $\alpha$ be a primitive element for $\F_q$, let $s_1, \dots, s_{\ell} \in \N$ be the indexes of the columns of $\mathcal{H}^\T$, selected to form the cycle, (we consider $s_i$ also if we select the shift of the $i$-th column) hence we have that $1\leq s_h\leq k$.

We can write $A$ in the following form:

\begin{equation*}\label{eq:generalA}
\small{\begin{bmatrix}
\alpha^{r_{1}s_1} & \alpha^{r_{2}s_2} & 0  & 0 &\cdots & 0 \\
0 & \alpha^{\left(r_{2}+i_1\right)s_2} & \alpha^{\left(r_{3}+i_1\right)s_3} & 0 & \cdots & 0 \\
0 & 0 & \alpha^{\left(r_{3} + i_1+i_2\right)s_3} & \alpha^{\left(r_{4} + i_1+i_2\right)s_4} &  \cdots & 0 \\
\vdots  & & & \ddots & \ddots &  \\
\alpha^{\left(r_{1}+i_1+ \dots+ i_{\ell-1}\right)s_1} & 0 & 0 & 0  & \cdots & \alpha^{\left(r_{\ell} +i_1+ \dots+ i_{\ell-1}\right)s_{\ell}}
\end{bmatrix}},
\end{equation*}
where $i_h\in\mathbb Z$ and $|i_h|$ is equal to a difference from $T_{s_{h+1}}$ for $h=1,\hdots\ell-1$ and $|i_1+ \dots+ i_{\ell-1}|$ is equal to a difference from $T_{s_1}$. Moreover, $1\leq r_h+i_1+\hdots+i_g\leq (\mu+1)(n-k)$ for $h=1,\hdots,\ell$ and $g=0,\hdots,\ell-1$.

We want to estimate the sufficient field size to have that this determinant is nonzero and therefore, we distinguish two cases.

\underline{\textbf{Case I:}} Assume that $\ell$ is odd. In this case, the determinant of a matrix of the form \eqref{eq:cycleA} is given by 
$$\det A = \prod_{\substack{i=1 \\ i \textnormal{ odd}}}^{2\ell} a_i + \prod_{\substack{i=1 \\ i \textnormal{ even}}}^{2\ell} a_i.$$

Hence, if the characteristic of the field is $p>2$, it is equal to $0$ in $\F_q$ if and only if 
\begin{align*} \alpha^{(i_1+i_2+\dots + i_{\ell-1})s_{1}} + \alpha^{i_1s_2 +i_2s_3 + \dots + i_{\ell-1}s_{\ell}}=0,
\end{align*}
which is equivalent to 
\begin{align*}
    (i_1+i_2+\dots + i_{\ell-1})s_{1}= i_1s_2 +i_2s_3 + \dots + i_{\ell-1}s_{\ell}+\frac{(q-1)}{2} \mod (q-1),
\end{align*}
and hence 
\begin{align*}
    i_1(s_2-s_1) + i_2(s_3-s_1) + \dots + i_{\ell-1}(s_{\ell}-s_1) -\frac{(q-1)}{2} = 0 \mod (q-1).
\end{align*}

It is then enough to consider $q$ bigger than the maximum value that can be reached by the function $$1+2\sum_{h=1}^{\ell-1}i_h(s_{h+1}-s_1).$$

Now, note that $i_h$ can be also negative but in general, we can say that $|i_h|\leq(\mu+1)(n-k)-1$. Moreover, $|s_i-s_1|\leq k-1$. Hence, if we can ensure that 
\begin{align*}
    q&>2((\mu+1)(n-k)-1)(\ell-1)(k-1)+1\\
    &=2(\mu+1)(n-k)(\ell-1)(k-1)-2(\ell-1)(k-1)+1,
\end{align*} 
with this construction we have a convolutional code whose sliding parity-check matrix is associated to a Tanner graph free from $2\ell$-cycles, with $\ell$ odd, not satisfying the FRC.

\begin{remark}
Observe that in Theorem \ref{thm:3x3minors}, we computed a more accurate estimation of the field size for getting rid of the $2\ell$ cycles, for $\ell = 3$, namely, $q>2(\mu+1)(n-k)(k-1)-2(k-1)+1$. The computation above shows that with $q>4(\mu+1)(n-k)(k-1)-4(k-1)+1$ we do not have $6$-cycles not satisfying the FRC. This difference is due to the possibility of a better estimation of the terms in the above inequality.
\end{remark}

With the discussion above we have proved the following result.

\begin{theorem}
Let $n,k,w$ be positive integers with $n>k$, $\T$ be a $(k,w)$-wDTS and $\C^\T$ be the $(n,k)_q$ convolutional code constructed from $\T$ with $q=p^N$ and $p>2$. A sufficient condition for obtaining a code whose sliding parity-check matrix is free from $2\ell$-cycles not satisfying the FRC with $\ell$ odd is to choose a field size $q>2(\mu+1)(n-k)(\ell-1)(k-1)-2(k-1)(\ell-1)+1$, where $\mu=\left\lceil\frac{m(\T)}{n-k}\right\rceil-1$ is the degree of the parity-check matrix of $\mathcal{C}^\T$.
\end{theorem}

\begin{example}
Consider again the code constructed in Example \ref{ex:Construction}. From Remark \ref{rem:cyclelength}, we know that the  highest length that we can have for a cycle is $10 = 2\cdot 5$, but for $q$ odd with $q>41$ all the $10$-cycles satisfy the FRC.
\end{example}

\underline{\textbf{Case II:}} Assume that $\ell$ is even. In this case, the determinant of a matrix of the form \eqref{eq:cycleA} is given by 
$$\det A = \prod_{\substack{i=1 \\ i \textnormal{ odd}}}^{2\ell} a_i - \prod_{\substack{i=1 \\ i \textnormal{ even}}}^{2\ell} a_i.$$

After some straightforward computation, it is easy to see that this determinant is equal to $0$ in $\F_q$ if and only if 
\begin{align*} \alpha^{(i_1+i_2+\dots + i_{\ell-1})s_{1}} = \alpha^{i_1s_2 +i_2s_3 + \dots + i_{\ell-1}s_{\ell}},
\end{align*}
which is equivalent to 
\begin{align*}
    (i_1+i_2+\dots + i_{\ell-1})s_{1}= i_1s_2 +i_2s_3 + \dots + i_{\ell-1}s_{\ell} \mod (q-1),
\end{align*}
and hence 
\begin{align*}
   f(i,s):= i_1(s_2-s_1) + i_2(s_3-s_1) + \dots + i_{\ell-1}(s_{\ell}-s_1) = 0 \mod (q-1).
\end{align*}
for $i:=(i_1,\dots, i_{\ell-1})$ and $s:=(s_1,\dots,s_\ell)$.

Moreover, we have the following constraints: 

\begin{enumerate}
    \item $ -(\mu+1)(n-k)+1\leq i_h \leq (\mu+1)(n-k)-1$ for $h=1,\dots,\ell-1$
    \item $-k+1\leq s_{h+1}-s_1\leq k-1$, for $h=1,\dots,\ell-1$;
\end{enumerate}

We have to find conditions on the corresponding wDTS to ensure that $f(i,s)$ is nonzero when viewed as an element of $\mathbb Z$ and then, we can determine a lower bound for $q$ in order that it is also nonzero modulo $q-1$.

Using Proposition \ref{dif}, we know that if none of the differences in the difference triangle set is divisible by $n-k$, then not all the values $s_1,\hdots,s_{\ell}$ can be identical. In particular, there is at least one $h\in\{2,\hdots,\ell\}$ such that $s_h-s_1\neq 0$.

\begin{theorem}
Let $\ell$ be an even integer, $k,n,w$ be integers such that $n>k$, $\T$ be a $(k,w)$-wDTS and $\mathcal{C}^\T$ be the $(n,k)_q$ convolutional code constructed from $\T$. Assume that $\T$ fulfills the conditions of Proposition \ref{dif} and has the property that $f(i,s)$ is nonzero in $\mathbb Z$ for all $s_1,\hdots,s_{\ell}\in\{1,\hdots,k\}$ not all equal if $|i_h|$ is equal to a difference from $T_{s_{h+1}}$ for $h=1,\hdots\ell-1$ and $|i_1+ \dots+ i_{\ell-1}|$ is equal to a difference from $T_{s_1}$ and $q>((\mu+1)(n-k)-1)\left((k-1)\frac{\ell}{2}+(k-2)\frac{\ell-2}{2}\right)+1$. Then, the Tanner graph associated to the sliding parity-check matrix of $\mathcal{C}^\T$ is free from $2\ell$-cycles that do not satisfy the FRC.
\end{theorem}

\begin{proof}
The conditions of the theorem ensure that $f(i,s)$ is nonzero in $\mathbb Z$. Moreover, it follows from Proposition \ref{dif} that $$((\mu+1)(n-k)-1)\left((k-1)\frac{\ell}{2}+(k-2)\frac{\ell-2}{2}\right)$$ is an upper bound for $|f(i,s)|$. Hence, the result follows.  
\end{proof}

Next, we want to give an example for a convolutional code that fulfills the conditions of the preceding theorem.

\begin{example}
Let $n=7$ and $k=2$ and $T_1=\{1,2,5,9\}$ and $T_2=\{1,2,4,10\}$, i.e. $\mu=1$. Note that $T_1$ is no difference triangle in the strict sense as $9-5=5-1$ but as $n-k=5$ does not divide $9-5$, we can still use it for the construction of our code (see Remark \ref{nk}). We get
$$\mathcal{H}^\T=\left[\begin{array}{cccccccccccccc}
\alpha & \alpha^2 & 1 & 0 & 0 & 0 & 0 & 0 & 0 & 0 & 0 & 0 & 0 & 0\\ 
\alpha^2 & \alpha^4 & 0 & 1 & 0 & 0 & 0 & 0 & 0 & 0 & 0 & 0 & 0 & 0\\ 
0 & 0 & 0 & 0& 1 & 0 & 0 & 0 & 0 & 0 & 0 & 0 & 0 & 0\\ 0 & \alpha^8 & 0 & 0 & 0 & 1 & 0 & 0 & 0 & 0 & 0 & 0 & 0 & 0\\ 
\alpha^5 & 0 & 0 & 0 & 0 & 0 &1 & 0 & 0 & 0 & 0 & 0 & 0 & 0\\ 
0 & 0 & 0 & 0 & 0 & 0 & 0 & \alpha & \alpha^2 & 1 & 0 & 0 & 0 & 0\\ 0 & 0 & 0 & 0 & 0 & 0 & 0 & \alpha^2 & \alpha^4 & 0 & 1 & 0 & 0 & 0\\ 
0 &  0 & 0 & 0 & 0 & 0 & 0 & 0 & 0 & 0 & 0& 1 & 0 & 0\\ 
\alpha^9 & 0 & 0 & 0 & 0 & 0 & 0 &  0 & \alpha^8 & 0 & 0 & 0 & 1 & 0\\ 
0 & \alpha^{20} & 0 & 0 & 0 & 0 & 0 & \alpha^5 & 0 & 0 & 0 & 0 & 0 &1
\end{array}\right].$$
Form Remark \ref{rem:cyclelength} one knows that with these parameters it is not possible to have cycles of length $2\ell$ for $\ell>4$. Moreover, from Theorem \ref{thm:2x2minors}, we obtain that we can exclude 4-cycles not fulfilling the FRC if $q>11$ and from Theorem \ref{thm:3x3minors}, that we can exclude 6-cycles not fulfilling the FRC if $q>19$ and $q$ is odd.
We will show that with the help of the preceding theorem, we can also exclude $8$-cycles in $\mathcal{H}^T$ that do not fulfill the FRC or in other words, all $2\ell $-cycles for any $\ell$ in $\mathcal{H}^\T$ fulfill the FRC for $q>19$. First, from Proposition \ref{dif}, we know that in the matrix $A$ representing any $8$-cycle we necessarily have $s_1=s_3$ and $s_2=s_4$ and each column of $\bar{H}^\T$ is involved once unshifted and once shifted by $5$. We get $f(i,s)=\pm (i_1+i_3)$ and have to exclude that $i_1\neq -i_3$. Considering $\mathcal{H}^\T$, we realize that 8-cycles are only possible for $s_1=s_3=1$ and $s_2=s_4=2$ and $i_1\in\{\pm 8,\pm 9\}$ and $i_3\in\{\pm 2,\pm 3\}$. Hence, for $q>9\cdot 2+1=19$, the corresponding convolutional code is free from $8$-cycles not fulfilling the FRC and hence, free from $2\ell$-cycles not fulfilling the FRC for any $\ell$.
\end{example}

To conclude this section, we will modify our construction from Definition \ref{Construction} in order to further relax the conditions on the underlying wDTS and still ensuring that we have no cycles not fulfilling the FRC. However, this will come with the cost of a larger field size.

\begin{definition}
Let $k,n$ be positive integers with $n>k$ and $\T:=\{T_1, \dots, T_{k}\}$ an $(k, w)$-wDTS with scope $m(\T)$. 
Set $\mu=\left\lceil\frac{m(\T)}{n-k}\right\rceil-1$ 
 and let $\alpha$ be a primitive element for $\F_q$. Moreover, let $P$ be a prime (with properties that will be determined later). For any $1\leq i \leq (\mu+1)(n-k)$, $1\leq l\leq k$, define
$$
\bar{H}^{(\T)}_{i,l} := \begin{cases}\alpha^{P^il} & \text{ if } i \in T_l \\
0 & \text{ otherwise}
\end{cases}.$$
\end{definition}

\begin{theorem}
Let $k,n,w$ be positive integers with $n>k$ and $\T$ be a DTS with $a_{i,1}=1$ for all $1\leq i\leq k$ and $\mathcal{C}$ be an $(n,k)_q$ convolutional code constructed from $\bar{H}^{(\T)}$. If $P>\ell k$ and $q>kP^{(\mu+1)(n-k)}\frac{P^{2\ell}-1}{P^{2\ell}-P^{2\ell-1}}+1$, then the Tanner graph associated to the sliding parity-check matrix contains no cycles of size $2\ell$ not fulfilling the FRC.
\end{theorem}

\begin{proof}
As with the construction from Definition \ref{Construction}, we obtain that $\det(A)=0$ if and only if a certain linear combination $\tilde{f}(i,s)$ of exponents of $P$ with $2\ell$ coefficients from $\{1,\hdots, k\}$ is zero. As the exponents correspond to row indices before a possible shift and the unshifted columns only intersect in the first row, all exponents that are equal to any other exponent are equal to 1. Moreover, as exponents from the same column of $A$ cannot be the same, at most $\ell$ exponents can be equal to 1. In summary, we obtain that $\tilde{f}(i,s)$ is of the form $\tilde{f}(i,s)=Px + P^{e_1}x_1+\cdots+P^{e_t}x_t$ with natural numbers $1<e_1<\cdots<e_t\leq (\mu+1)(n-k)$, $t\in\{\ell,\hdots,2\ell\}$, $x_j\in\{-k,\hdots,k\}\setminus\{0\}$ for $j=1,\hdots,t$ and $x\in\{-\ell k,\hdots,+\ell k\}$. Since $m$ was chosen to be a prime larger than $\ell k$, $\tilde{f}(i,s)$ is nonzero in $\mathbb Z$. Furthermore, $|\tilde{f}(i,s)|\leq k\sum_{i=0}^{2\ell-1}P^{(\mu+1)(n-k)-i}=kP^{(\mu+1)(n-k)}\frac{P^{2\ell}-1}{P^{2\ell}-P^{2\ell-1}}$ and hence it cannot be zero modulo $q-1$.  
\end{proof}

Finally, we illustrate our modified construction with an example.

\begin{example}
If we take the DTS $\T=\{\{1,2,5\},\{1,3,8\}\}$ to construct an $(6,2)_q$ convolutional code, we have $m(\T)=8$ and $\mu=1$. If we want that the girth of the corresponding parity-check matrix is at least 12, we have to choose $P>10$, i.e. $P=11$. To get the desired property it would be sufficient if the field size is larger than $4.716\times 10^8$. If it is sufficient to have a girth of at least 8, it would be enough to choose $P=7$ and the sufficient field size decreases to $1.35\times 10^7$.
\end{example}

\section{Acknowledgements}
The authors acknowledge the support of  Swiss National Science Foundation grant n. 188430. Julia Lieb acknowledges also the support of the German Research Foundation grant LI 3101/1-1.

\bibliographystyle{abbrv}

\bibliography{references}

\begin{thebibliography}{10}

\bibitem{alfarano2020construction}
G.~N. {Alfarano}, J.~{Lieb}, and J.~{Rosenthal}.
\newblock Construction of rate $(n - 1 )/n$ non-binary {LDPC} convolutional
  codes via difference triangle sets.
\newblock In {\em 2020 IEEE International Symposium on Information Theory
  (ISIT)}, pages 138--143, 2020.

\bibitem{amiri2014analysis}
B.~Amiri, J.~Kliewer, and L.~Dolecek.
\newblock Analysis and enumeration of absorbing sets for non-binary graph-based
  codes.
\newblock {\em IEEE Transactions on Communications}, 62(2):398--409, 2014.

\bibitem{bates2005low}
S.~Bates, Z.~Chen, and X.~Dong.
\newblock Low-density parity-check convolutional codes for ethernet networks.
\newblock In {\em PACRIM. 2005 IEEE Pacific Rim Conference on Communications,
  Computers and signal Processing, 2005.}, pages 85--88. IEEE, 2005.

\bibitem{bates2006termination}
S.~Bates, D.~G. Elliott, and R.~Swamy.
\newblock Termination sequence generation circuits for low-density parity-check
  convolutional codes.
\newblock {\em IEEE Transactions on Circuits and Systems I: Regular Papers},
  53(9):1909--1917, 2006.

\bibitem{battaglioni2019girth}
M.~Battaglioni, M.~Baldi, F.~Chiaraluce, and M.~Lentmaier.
\newblock Girth properties of time-varying {SC-LDPC} convolutional codes.
\newblock In {\em 2019 IEEE International Symposium on Information Theory
  (ISIT)}, pages 2599--2603. IEEE, 2019.

\bibitem{berrou1993near}
C.~Berrou, A.~Glavieux, and P.~Thitimajshima.
\newblock Near {S}hannon limit error-correcting coding and decoding:
  {T}urbo-codes. 1.
\newblock In {\em Proceedings of ICC'93-IEEE International Conference on
  Communications}, volume~2, pages 1064--1070. IEEE, 1993.

\bibitem{chee1997constructions}
Y.~M. Chee and C.~J. Colbourn.
\newblock Constructions for difference triangle sets.
\newblock {\em IEEE Transactions on Information Theory}, 43(4):1346--1349,
  1997.

\bibitem{chen1992disjoint}
Z.~Chen, P.~Fan, and F.~Jin.
\newblock Disjoint difference sets, difference triangle sets, and related
  codes.
\newblock {\em IEEE Transactions on Information Theory}, 38(2):518--522, 1992.

\bibitem{chung2001analysis}
S.-Y. Chung, T.~J. Richardson, and R.~L. Urbanke.
\newblock Analysis of sum-product decoding of low-density parity-check codes
  using a gaussian approximation.
\newblock {\em IEEE Transactions on Information theory}, 47(2):657--670, 2001.

\bibitem{colbourn1996difference}
C.~J. Colbourn.
\newblock Difference triangle sets.
\newblock {\em Chapter in The CRC Handbook of Combinatorial Designs by CJ
  Colbourn and J. Dintz}, pages 312--317, 1996.

\bibitem{davey1998low}
M.~C. Davey and D.~J. MacKay.
\newblock Low density parity check codes over {GF} ($q$).
\newblock In {\em 1998 Information Theory Workshop (Cat. No. 98EX131)}, pages
  70--71. IEEE, 1998.

\bibitem{davey1998monte}
M.~C. Davey and D.~J. MacKay.
\newblock Monte {C}arlo simulations of infinite low density parity check codes
  over {GF}($q$).
\newblock In {\em Proc. of Int. Workshop on Optimal Codes and related Topics},
  pages 9--15. Citeseer, 1998.

\bibitem{forney1970}
G.~Forney.
\newblock Convolutional codes {I}: {A}lgebraic structure.
\newblock {\em IEEE Transactions on Information Theory}, 16(6):720--738, 1970.

\bibitem{gallager1962low}
R.~Gallager.
\newblock Low-density parity-check codes.
\newblock {\em IRE Transactions on Information Theory}, 8(1):21--28, 1962.

\bibitem{gluesing2006strongly}
H.~Gluesing-Luerssen, J.~Rosenthal, and R.~Smarandache.
\newblock Strongly-{MDS} convolutional codes.
\newblock {\em IEEE Transactions on Information Theory}, 52(2):584--598, 2006.

\bibitem{johnson2001regular}
S.~J. Johnson and S.~R. Weller.
\newblock Regular low-density parity-check codes from combinatorial designs.
\newblock In {\em Proceedings 2001 IEEE Information Theory Workshop (Cat. No.
  01EX494)}, pages 90--92. IEEE, 2001.

\bibitem{kailath1980}
T.~Kailath.
\newblock {\em Linear systems}, volume 156.
\newblock Prentice-Hall Englewood Cliffs, NJ, 1980.

\bibitem{klove1989bounds}
T.~Klove.
\newblock Bounds and construction for difference triangle sets.
\newblock {\em IEEE Transactions on Information Theory}, 35(4):879--886, 1989.

\bibitem{kou2000construction}
Y.~Kou, S.~Lin, and M.~Fossorier.
\newblock Construction of low density parity check codes: a geometric approach.
\newblock In {\em Proceedings of the 2nd International Symposium on Turbo Codes
  and Related Topics}, pages 137--140, 2000.

\bibitem{kou2001low}
Y.~Kou, S.~Lin, and M.~P. Fossorier.
\newblock Low-density parity-check codes based on finite geometries: a
  rediscovery and new results.
\newblock {\em IEEE Transactions on Information theory}, 47(7):2711--2736,
  2001.

\bibitem{mackay1996near}
D.~J. MacKay and R.~M. Neal.
\newblock Near shannon limit performance of low density parity check codes.
\newblock {\em Electronics letters}, 32(18):1645--1646, 1996.

\bibitem{massey1963threshold}
J.~L. Massey.
\newblock Threshold decoding.
\newblock 1963.

\bibitem{oswald2002capacity}
P.~Oswald and A.~Shokrollahi.
\newblock Capacity-achieving sequences for the erasure channel.
\newblock {\em IEEE Transactions on Information Theory}, 48(12):3017--3028,
  2002.

\bibitem{poulliat2008design}
C.~Poulliat, M.~Fossorier, and D.~Declercq.
\newblock Design of regular $(2, d_c)$-{LDPC} codes over {GF}($q$) using their
  binary images.
\newblock {\em IEEE Transactions on Communications}, 56(10):1626--1635, 2008.

\bibitem{pusane2011deriving}
A.~E. Pusane, R.~Smarandache, P.~O. Vontobel, and D.~J. Costello.
\newblock Deriving good {LDPC} convolutional codes from {LDPC} block codes.
\newblock {\em IEEE Transactions on Information Theory}, 57(2):835--857, 2011.

\bibitem{richardson2001capacity}
T.~J. Richardson and R.~L. Urbanke.
\newblock The capacity of low-density parity-check codes under message-passing
  decoding.
\newblock {\em IEEE Transactions on information theory}, 47(2):599--618, 2001.

\bibitem{robinson1967class}
J.~P. Robinson and A.~Bernstein.
\newblock A class of binary recurrent codes with limited error propagation.
\newblock {\em IEEE Transactions on Information Theory}, 13(1):106--113, 1967.

\bibitem{ro01}
J.~Rosenthal.
\newblock Connections between linear systems and convolutional codes.
\newblock In {\em Codes, Systems, and Graphical Models}, pages 39--66.
  Springer, 2001.

\bibitem{rosenthal2000constructions}
J.~Rosenthal and P.~O. Vontobel.
\newblock Constructions of {LDPC} codes using {R}amanujan graphs and ideas from
  {M}argulis.
\newblock In {\em in Proc. of the 38-th Allerton Conference on Communication,
  Control, and Computing}. Citeseer, 2000.

\bibitem{tanner1981recursive}
R.~Tanner.
\newblock A recursive approach to low complexity codes.
\newblock {\em IEEE Transactions on information theory}, 27(5):533--547, 1981.

\bibitem{tanner1981error}
R.~M. Tanner.
\newblock Error-correcting coding system, Oct.~13 1981.
\newblock US Patent 4,295,218.

\bibitem{tanner1987convolutional}
R.~M. Tanner.
\newblock {\em Convolutional codes from quasi-cyclic codes: A link between the
  theories of block and convolutional codes}.
\newblock University of California, Santa Cruz, Computer Research Laboratory,
  1987.

\bibitem{tanner2004ldpc}
R.~M. Tanner, D.~Sridhara, A.~Sridharan, T.~E. Fuja, and D.~J. Costello.
\newblock Ldpc block and convolutional codes based on circulant matrices.
\newblock {\em IEEE Transactions on Information Theory}, 50(12):2966--2984,
  2004.

\bibitem{tong1970systematic}
S.-Y. Tong.
\newblock Systematic construction of self-orthogonal diffuse codes.
\newblock {\em IEEE Transactions on Information Theory}, 16(5):594--604, 1970.

\bibitem{tong1971}
S.-Y. Tong.
\newblock Character-correcting convolutional self-orthogonal codes.
\newblock {\em Information and Control}, 18:183--202, 1971.

\bibitem{vasic2001combinatorial}
B.~Vasic and O.~Milenkovic.
\newblock Combinatorial constructions of structured low-density parity check
  codes for iterative decoding.
\newblock {\em IEEE Trans. Inform. Theory}, 2001.

\bibitem{vasic2004combinatorial}
B.~Vasic and O.~Milenkovic.
\newblock Combinatorial constructions of low-density parity-check codes for
  iterative decoding.
\newblock {\em IEEE Transactions on information theory}, 50(6):1156--1176,
  2004.

\bibitem{wiberg1996codes}
N.~Wiberg.
\newblock {\em Codes and decoding on general graphs}.
\newblock PhD thesis, 1996.

\bibitem{zhou2010cycle}
H.~Zhou and N.~Goertz.
\newblock Cycle analysis of time-invariant {LDPC} convolutional codes.
\newblock In {\em 2010 17th International Conference on Telecommunications},
  pages 23--28. IEEE, 2010.

\end{thebibliography}

\end{document}